\documentclass{amsart}
\usepackage{color}
\usepackage{cite}
\newtheorem{theorem}{Theorem}[section]
\newtheorem{corollary}[theorem]{Corollary}

\newtheorem{proposition}[theorem]{Proposition}

\theoremstyle{definition}
\newtheorem{definition}[theorem]{Definition}

\theoremstyle{remark}
\newtheorem{remark}[theorem]{Remark}
\numberwithin{equation}{section}

\newcommand{\A}{\mathfrak{A}}
\newcommand{\HI}{\mathfrak{H}}
\newcommand{\R}{\mathbb{R}}
\newcommand{\As}{\mathfrak{A}_s}

\newcommand{\af}{\mathfrak{a}}
\newcommand{\ns}{\mathfrak{n-s}}

\begin{document}
\title[On the dimensions of oscillator-like algebras]{On the dimensions of oscillator-like algebras induced by orthogonal polynomials: NON-SYMMETRIC CASE}
\author{G. Honnouvo$^1$ and K. Thirulogasanthar$^2$}
\address{$^1$ Department of Mathematics and Statistics, McGill University,  
   805 Sherbrooke Street West, Montreal, Quebec H3A 2K6, Canada }
\email{g\_honnouvo@yahoo.fr}
\address{$^2$Department of Computer Science and Software Engineering, Concordia University,  1455 de Maisonneuve Blvd. West, Montreal, Quebec, H3G 1M8, Canada }
\email{santhar@gmail.com}
\subjclass{Primary 33C45, 33C80, 33D80}
\date{\today}
\keywords{Orthogonal polynomials, oscillator-like algebras, deformed oscillator algebras.}
\begin{abstract}
There is a generalized oscillator-like algebra associated with every class of orthogonal polynomials, on the real line, satisfying a four term non-symmetric recurrence relation.
This note presents necessary and sufficient conditions, on the coefficients of the recurrence relation, for such algebras to be of finite dimension. As examples, we discuss the dimensions of oscillator-like algebras associated with Laguerre and Jacobi polynomials.
\end{abstract}
\maketitle
\pagestyle{myheadings}
\section{Introduction}
The usual harmonic oscillator annihilation , creation and the number operators are defined respectively as
\begin{equation}\label{har}
\af\Psi_n=\sqrt{n}\Psi_{n-1}; n\geq 1,\quad \af^{\dagger}\Psi_n=\sqrt{n+1}\Psi_{n+1}; n\geq 0,\quad N=\af^{\dagger}\af
\end{equation}
and $\af\Psi_0=0$, where $\{\Psi_n\}_{n=0}^{\infty}$ is an orthonormal basis of the harmonic oscillator Fock space. In this case 
$$[\af, \af^{\dagger}]=I,\quad [N,\af]=-\af,\quad [N,\af^{\dagger}]=\af^{\dagger},\quad (\af^{\dagger})^{\dagger}=\af,\quad N^{\dagger}=N$$
and the algebra generated by $\{I, \af, \af^{\dagger}, N\}$ is the usual Weyl-Heisenberg algebra. We call this algebra $\mathfrak{A}_{WH}$. It is well-known that the dimension of this algebra is four.

Several generalizations and deformations of the algebra $\mathfrak{A}_{WH}$ have been studied in the literature, for example, \cite{Bec1, Bec2, key3, B, qd, Bur, BH, FV}. In generalizing or deforming the algebra $\mathfrak{A}_{WH}$ we inclined to stay as close as to the commutation relations of the algebra $\mathfrak{A}_{WH}$. In the following we shall provide conditions, in terms of the coefficients of some recurrence relations satisfied by the Fock basis of generalized oscillator-like algebras, for such algebras  to be of the same dimension as the algebra $\mathfrak{A}_{WH}$	.

In particular, in the following we shall discuss the dimensions of generalized oscillator algebras presented in \cite{B} and the dimensions of a modified version of an oscillator-like algebra presented in \cite{Bec1}.

In a recent paper, \cite{Hon}, we have considered the dimension of generalized oscillator algebras associated with orthogonal polynomials, on the real line, that are orthogonal with respect to a symmetric probability measure.

Let $\HI=L^2(\R, d\mu),$ where $\mu$ is a probability measure on $\R$ with finite moments
\begin{equation}\label{E1}
\mu_n=\int_{-\infty}^{\infty}x^nd\mu(x).
\end{equation}
These moments uniquely define the real sequences $\{a_n\}_{n=0}^{\infty},~\{b_n\}_{n=0}^{\infty}$ and the system of orthogonal polynomials $\{\Psi_n(x)\}_{n=0}^{\infty}$ satisfying the recurrence relation \cite{B}
\begin{equation}\label{E2}
x\Psi_n(x)=b_n\Psi_{n+1}(x)+a_n\Psi_n(x)+b_{n-1}\Psi_{n-1}(x),~~ \Psi_0(x)=1,~b_{-1}=0;~~~n=0,1,2,\cdots.
\end{equation}
The polynomials (normalized) $\{\Psi_n(x)\}_{n=0}^{\infty}$ form an orthonormal basis for a Fock space associated with a generalized oscillator algebra provided that $b_n$'s and $\mu_n$'s are connected by a specific  relation \cite{B}. There are two cases associated with (\ref{E2}) \cite{B, BD2, BD3}:
\begin{enumerate}
\item[(i)] $a_n=0$, symmetric case
\item[(ii)] $a_n\not=0,$ non-symmetric case
\end{enumerate}
The primary aim of this article is to investigate the dimension of an oscillator-like algebra obeying the recurrence relation (\ref{E2}).  We shall provide necessary and sufficient conditions, in terms of $a_n$ and $b_n$ of (\ref{E2}), for such an oscillator-like algebra to be of finite dimension. This result, in a manner, can be viewed as a dimension wise classification for such algebras.

The rest of the article is organized as follows. In subsection 2.1 we briefly discuss the symmetric case. In particular we shall respond to the claims made in \cite{BD2, BD3} about the results of our earlier paper \cite{Hon}. In section 2.2 we discuss the non-symmetric case and also comment on the results provided in \cite{BD2,BD3} about the oscillator algebra associated with the non-symmetric case. Subsection 2.3 deals with oscillator-like algebras obeying the recurrence relation (\ref{E2}). In section 3 we present the main results of this manuscript. That is, we present a necessary and sufficient condition on $a_n$ and $b_n$ of (\ref{E2}) for oscillator-like algebras to be of finite dimension.
Some examples accommodating our claim are presented in section 4. Section 5 ends the manuscript with a conclusion.

\section{Classes of generalized oscillator and oscillator-like algebras}
In this section we shall provide a class of generalized oscillator and oscillator-like algebras based on \cite{B, Bec1, Bec2}. In particular we shall respond to the claims made in \cite{BD2, BD3} about our earlier paper \cite{Hon}.
\subsection{Symmetric case}
 Let $\mu$ be a symmetric probability measure on the real line, $\mathbb{R}$. That is, the measure $\mu$ satisfies
\begin{equation}\label{E3}
\int_{-\infty}^{\infty} \mu(dx)=1,\quad \text{and}\quad \mu_{2k+1}=\int_{-\infty}^{\infty}x^{2k+1} \mu(dx)=0;\quad k=0,1,...
\end{equation}
Let 
\begin{equation}\label{sb}
\{ b_n \}_{n=0}^\infty,\quad b_n>0; \quad n=0, 1,... 
\end{equation}
be a positive sequence defined by the algebraic equations system
\begin{equation}\label{a-1}
\sum_{m=0}^{[\frac{1}{2}n]}\sum_{s=0}^{[\frac{1}{2}n]}(-1)^{m+s}\alpha_{2m-1,n-1}\alpha_{2s-1,n-1}\frac{\mu_{2n-2m-2s+2}}{(b_{n-1}^2)!}= b_{n-1}^2 +b_{n}^2;\quad n=0,1,2,...,
\end{equation}
where $(b_{n-1}^2)!= b_{0}^2b_{1}^2...b_{n-1}^2,$ the integral part of $a$ is denoted by $[a]$, and the coefficients $\alpha_{ij}$ are given by
\begin{eqnarray}
\alpha_{2p-1,n-1}= \sum_{k_1=2p-1}^{n-1}b_{k_1}^2\sum_{k_2=2p-3}^{k_1-2}b_{k_2}^2
...\sum_{k_p=1}^{k_{p-1}-2}b_{k_p}^2.
\end{eqnarray}
In (2.3) and (2.4), for $k=0,1,\cdots$, $\mu_{2k}$ are known and it determine the sequence $\{b_n\}$ and the sequence $\{b_n\}$ determines $\alpha_{n,m}$ through (2.4). In order to get the orthogonality of the polynomials, $\Psi_n$, in (\ref{Re}), $\mu_{2k}$ and $b_n$ have to be related by the relation (\ref{a-1}) (see Theorem \ref{day}).

Let us consider a system $\{\Psi_n(x) \}_{n=0}^\infty$ of polynomials defined by the recurrence relations $(n\geq 0):$
\begin{equation}\label{Re}
x\Psi_n(x)=b_n\Psi_{n+1}(x)+ b_{n-1}\Psi_{n-1}(x),\quad \Psi_0(x)=1,\quad b_{-1}=0,
\end{equation}
where $\{b_n \}_{n=0}^\infty$ is a given positive sequence satisfying the relation $(\ref{a-1})$.
The following theorem was proved in \cite{B}.
\begin{theorem}\label{day}
The polynomial system $\{\Psi_n(x) \}_{n=0}^\infty$ is orthonormal in the Hilbert space $\mathfrak{H}$ if and only if the coefficients $b_n$ and the moments $\mu_{2k}$ are connected by relation (\ref{a-1}).
\end{theorem}
Let $\{\Psi_n(x) \}_{n=0}^\infty$ be an orthonormal basis of the Fock space $\mathcal{H}_s$ which satisfies the recurrence relation (\ref{Re}). That is,
$$\mathcal{H}_s=\overline{\text{span}}\left\{ \Psi_n(x)~|~n=0,1,2,...\right\}\subseteq\mathfrak{H},$$
where the bar stands for the closure of the linear span taken under the norm topology of $\mathfrak{H}$. Define the ladder operators $\af_s^\dagger$, $\af_s$ and the number operator $N_s$ in the Fock space, $\mathcal{H}_s$, by the usual formulas:
\begin{eqnarray}\label{b-1}
\af_s^\dagger\Psi_n(x)&=&\sqrt{2}b_n \Psi_{n+1}(x),\nonumber\\ \af_s\Psi_n(x)&=&\sqrt{2}b_{n-1} \Psi_{n-1}(x),\\
 N_s\Psi_n(x)&=&n\Psi_n(x).\nonumber
 \end{eqnarray}
It can be readily seen that $(\af_s^{\dagger})^{\dagger}=\af_s$. The polynomial set$\{\Psi_n(x) \}_{n=0}^\infty$ is called a canonical polynomial system when it is defined by the recurrence relation (\ref{Re}). The canonical polynomial system $\{\Psi_n(x) \}_{n=0}^\infty$ is uniquely determined by the symmetric probability measure $\mu$.
Now, as usual, let the position operator be
\begin{equation}\label{P}
Q_s=\frac{\af_s+\af_s^{\dagger}}{\sqrt{2}}.
\end{equation}
In order to guarantee
\begin{equation}\label{pos}
Q_s\Psi_n(x)=x\Psi_n(x)
\end{equation}
the symmetry of the measure is required  and the relation (\ref{pos}) is essential for the three term recursion relation (\ref{Re}). In fact, the relation (\ref{pos}) provides the connection between the operators $\af_s, \af_s^{\dagger}$ and the recurrence relation (\ref{Re}) \cite{B, Hon}. 
The operator $N_s$ is self adjoint in the Fock space. Therefore for any Borel function $B$, through the spectral theorem \cite{Oli}, one can define the operator $B(N_s)$. In this regard, we take a function
 $B(N_s)$ of operator $N_s$ in the space $\mathcal{H}_s$ which acts on the basis vectors, $\{\Psi_n(x) \}_{n=0}^\infty$ as 
\begin{equation}\label{b3}
B(N_s)\Psi_n(x)=b_{n-1}^2\Psi_n(x),\quad\text{and}\quad B(N_s+I_s)\Psi_n(x)=b_n^2\Psi_n(x);\quad n\geq 0,
\end{equation}
where $I_s$ is the identity operator on $\mathcal{H}_s$. The following result is proved in \cite{B, BD}:
\begin{theorem}\label{Th1}
The operators $\af_s, \af_s^{\dagger}$ and $N_s$ obey the following commutation relations
\begin{equation}\label{Com}
[\af_s, \af_s^\dagger]= 2\left(B(N_s+I_s)-B(N_s)\right), \quad [N_s, \af_s^\dagger]=\af_s^\dagger,\quad [N_s, \af_s]=-\af_s.
\end{equation}
\end{theorem}
\begin{definition}\label{DD1}
An algebra $\mathcal A_s$ is called a generalized oscillator algebra corresponding to the orthonormal system $\{\Psi_n(x) \}_{n=0}^\infty$, which satisfies (\ref{Re}), if $\mathcal A_s$ is generated by the operators $\af_s^\dagger$, $\af_s$, $N_s$, and $I_s$ and by their commutators. These operators should also satisfy the relations (\ref{b-1}) and (\ref{Com}).  
\end{definition}

At this point we like to emphasize a word about the above definition.  The algebra $\mathfrak{A}_s$ consists the operators $\af_s, \af_s^{\dagger}, N_s$ and $I_s$ and their repeated commutators only.  Further, the above oscillator algebra arises only in the symmetric case \cite{B, Hon}. Further, the algebra $\mathfrak{A}_s$ may be considered as a generalization of the algebra $\mathfrak{A}_{WH}$.

Regarding the dimension of the algebra $\mathfrak{A}_s$ we have proved the following result in \cite{Hon}.
\begin{theorem}\label{Th2}
The generalized oscillator algebra $\mathfrak{A}_s$ is of finite dimension if and only if
\begin{equation}\label{bn}
b_n^2=\alpha_0+\alpha_1n+\alpha_2n^2,\quad b_{-1}=0,\quad \alpha_0, \alpha_1, \alpha_2\in\R 
\end{equation}
and in this case the dimension of the algebra is four.
\end{theorem}
\begin{remark}\label{Rem1}
At this point we like to respond to the comments made in \cite{BD2} and \cite{BD3} about our earlier paper \cite {Hon}.  In \cite{BD2} and \cite{BD3} the authors claimed that the sufficient part of the Theorem (\ref{Th2}) is incorrect.  They indicated that for the sufficient part to be true, in addition to (\ref{bn}), the coefficients $\alpha_0, \alpha_1$ and $\alpha_2$ of (\ref{bn}) must also satisfy the relation
\begin{equation}\label{bor}
\alpha_1=\alpha_0+\alpha_2
\end{equation}
However, the relation (\ref{bor}) is indeed included in Theorem (\ref{Th2}). It can be easily seen that in the recurrence relation (\ref{Re}) we have $b_{-1}=0$ and in (\ref{bn}) if $b_{-1}=0$ then we obtain $\alpha_1=\alpha_0+\alpha_2$. In this regard, Theorem (\ref{Th2}) is correct in its own form.
\end{remark}
\subsection{Non-symmetric case}
For a symmetric probability measure a non-symmetric recurrence relation can be transformed to a symmetric one. For details we refer the reader to Section 5 in \cite{B}. Also the following is extracted from \cite{B} as needed here.\\

Let $\mu$ be a probability but not necessarily a symmetric measure on $\R$. Let $\HI_{\ns}=L^2(\R, \mu)$ and
\begin{equation}\label{mom}
\mu_0=1,\quad \mu_k=\int_{-\infty}^{\infty}x^k d\mu(x);\quad k=1,2,\cdots.
\end{equation}
Let the real sequences $\{a_n\}_{n=0}^{\infty}, \{b_n\}_{n=0}^{\infty}$ be solutions of the system
\begin{equation}\label{sys}
\left\{\begin{array}{cc}
&A_{k,n}=b_nA_{k-1,n+1}+a_nA_{k-1,n}+b_{n-1}A_{k-1,n-1};\quad n= 0,1,2,\cdots; b_{-1}=0\\
&A_{0,0}=1,~~ A_{k,0}=\mu_k,~~ A_{0,k}=0;\quad  k\geq 1.
\end{array}\right.
\end{equation}
There is a unique solution to the system (\ref{sys}) with respect to the variables $(a_n, b_n, A_{k,n});~~n\geq 0, k\geq 0$. That is, the initial conditions $A_{0,0}=1, A_{k,0}=\mu_k, A_{0,k}=0$ are given for $k\geq 1$.  Through (2.14) we look for the real sequences $\{b_n\}_{n=0}^{\infty}$ and $\{a_n\}_{n=0}^{\infty}$.\\
If sequences $\{a_n\}_{n=0}^{\infty}, \{b_n\}_{n=0}^{\infty}$ are given, then one can define the polynomial system by the recurrence relation
\begin{equation}\label{nRe}
x\Psi_n(x)=b_n\Psi_{n+1}(x)+a_n\Psi_n(x)+b_{n-1}\Psi_{n-1}(x);~~n\geq 0, b_{-1}=0, \Psi_0(x)=1.
\end{equation}
\begin{theorem}\cite{B}Let $\{\Psi_n(x)\}_{n=0}^{\infty}$ be a real polynomial system defined by (\ref{nRe}) and let $\mu$ be a probability measure on $\R$. The system of polynomials $\{\Psi_n(x)\}_{n=0}^{\infty}$ is orthonormal with respect to the measure $\mu$ on $\R$ if and only if the coefficients $\{a_n\}_{n=0}^{\infty}, \{b_n\}_{n=0}^{\infty}$ involved in (\ref{nRe}) are the solutions of the system (\ref{sys}), where $\mu_k$ are defined by (\ref{mom}).
\end{theorem}
Let $\{\Psi_n(x)\}_{n=0}^{\infty}$ be an orthonormal system satisfying the four term recurrence relation (\ref{nRe}) and let
$$\mathcal{H}_{\ns}=\overline{\mbox{span}}\{\Psi_n(x)~\vert~n=0,1,2,\cdots\}.$$
On $\mathcal{H}_{\ns}$ define the operators $Q_{\ns}$ and $P_{\ns}$ as
\begin{eqnarray*}
Q_{\ns}\Psi_n(x)&=&b_{n-1}\Psi_{n-1}(x)+a_n\Psi_n(x)+b_n\Psi_{n+1}(x)\\
P_{\ns}\Psi_n(x)&=&i[b_{n-1}\Psi_{n-1}(x)-b_n\Psi_{n+1}(x)]+a_n\Psi_n(x);\quad n\geq 0.
\end{eqnarray*}
Let
\begin{equation}\label{npos}
\widehat{Q}_{\ns}=\mbox{Re}(Q_{\ns}-P_{\ns}),\quad \widehat{P}_{\ns}=-i\mbox{Im}(Q_{\ns}-P_{\ns}).
\end{equation}
Define the ladder operators as
\begin{equation}\label{nlad}
\widehat{\af}_{\ns}=\frac{\widehat{Q}_{\ns}-i\widehat{P}_\ns}{\sqrt{2}},\quad 
\widehat{\af}_{\ns}^{\dagger}=\frac{\widehat{Q}_{\ns}+i\widehat{P}_\ns}{\sqrt{2}}.
\end{equation}
Also take $N_{\ns}=N_s$ and $I_\ns$ to be the identity operator on $\mathcal{H}_\ns$.
Then for the operators $\widehat{\af}_{\ns}, \widehat{\af}_{\ns}^{\dagger}, N_\ns$ and $I_\ns$ the formulas (\ref{b-1}) and Theorem (\ref{Th1}) are valid. Let the algebra generated in this case be $\mathfrak{A}_\ns$. In the non-symmetric case the position operator, $\widehat{Q}_{\ns}$, does not have to be an operator of the multiplication by an independent variable. However, since the orthonormal system $\{\Psi_n(x)\}_{n=0}^{\infty}$ satisfies the recurrence relation (\ref{nRe}), by the definition of $Q_{\ns}$, the operator $Q_{\ns}$ is an operator of multiplication by an independent variable.
\begin{remark}\label{Rem2}
Since the sets of operators $$\{\af_s, \af_s^{\dagger}, N_s, I_s\}\quad\mbox{and}\quad \{\widehat{\af}_\ns, \widehat{\af}_\ns^{\dagger}, N_\ns, I_\ns\}$$ are defined by the same relations (\ref{b-1}) and satisfy the same commutation relations (Theorem \ref{Th1}), in \cite{BD2, BD3} the authors claimed that the algebras $\As$ and $\mathfrak{A}_{\ns}$ coincide and therefore Theorem (\ref{Th2}) is valid for the algebra $\mathfrak{A}_\ns$ as well. According to the definition of the algebra $\mathfrak{A}_\ns$ given in \cite{BD2, BD3} the authors claim is true. However, in the symmetric case the operator $Q_s$ is an operator of multiplication by an independent variable while the operator $\widehat{Q}_{\ns}$ does not hold this property.
\end{remark}
\subsection{Oscillator-like algebras}
 A particular kind of deformation to the creation operator of the algebra $\mathfrak{A}_{WH}$ is proposed in \cite{Bec1} and then used, for example, in \cite{Bec2, Rok}. In \cite{Bec1} the authors deformed the creation operator as
\begin{equation}\label{defo}
\af_{\lambda}^{\dagger}=\af^{\dagger}+\lambda I;\quad\lambda\in\R,
\end{equation}
where $\lambda$ is a real continuous parameter, without changing the annihilation operator. Even after the deformation the commutation relations of the Weyl-Heisenberg algebra remain unchanged, that is,
\begin{equation}\label{wh}
[\af, \af_{\lambda}^{\dagger}]=I,\quad [N_{\lambda},\af]=-\af,\quad [N_{\lambda},\af_{\lambda}^{\dagger}]=\af_{\lambda}^{\dagger},
\end{equation}
where $N_{\lambda}=\af_{\lambda}^{\dagger}\af$. However, $(\af_{\lambda}^{\dagger})^{\dagger}\not=\af$ and $N_{\lambda}^{\dagger}\not=N_{\lambda}$. In \cite{Bec1, Bec2} the authors called the algebra generated by $\{\af, \af_{\lambda}^{\dagger}, N_{\lambda}, I\}$ an oscillator-like algebra. In the following we propose a different, however similar to (\ref{defo}), deformation to the creation operator of the symmetric case and obtain an oscillator-like algebra. In fact, in (\ref{defo}) we replace the real parameter $\lambda$ by the $a_n$'s of the recurrence  relation (\ref{nRe}) and the operator $\af$ by the operator $\af_s$.\\

Let $\{a_n\}_{n=0}^{\infty}$, $\{b_n\}_{n=0}^{\infty}$  be the sequence of real numbers appearing in the four term recurrence relation (\ref{nRe}) and $\{\Psi_n(x)\}_{n=0}^{\infty}$ be an orthonormal polynomials system satisfying the recurrence relation (\ref{nRe}). Let $$\mathcal{H}=\overline{\mbox{span}}\{\Psi_n(x)~\vert~n=0,1,2,\cdots\}.$$ In $\mathcal{H}$, define the operator $D $ by
\begin{equation}\label{D}
D\Psi_n=\sqrt{2}a_n\Psi_n,\quad n=0,1,...
 \end{equation}
Also on $\mathcal{H}$ define the operators
\begin{equation}
A=\af_s,\quad A^\dagger=\af_s^\dagger+ D,\quad N=N_s
 \end{equation}
and $I$, the identity operator on $\mathcal{H}$.
Then their actions take the form
\begin{equation}\label{A-1}
A^\dagger\Psi_n=\sqrt{2}b_n \Psi_{n+1}+\sqrt{2}a_n \Psi_{n},\quad A\Psi_n=\sqrt{2}b_{n-1} \Psi_{n-1}.
 \end{equation}
We also have
\begin{equation}\label{wh1}
[A, A^{\dagger}]=2(B(N+I)-B(N))+Af(N),\quad [N,A]=-A,\quad [N,A^{\dagger}]=\af_s^{\dagger},
\end{equation}
where $f(N)$ is a function of the self-adjoint operator $N$ acting as $$f(N)\Psi_n(x)=\sqrt{2}(a_n-a_{n-1})\Psi_n(x).$$ Once again $(A^{\dagger})^{\dagger}\not=A$.
Let $\mathfrak{A}$ be the oscillator-like algebra generated by $\{I, A, A^{\dagger}, N\}$.
 Now, as usual, let the position operator be
\begin{equation}\label{thepos}
Q=\frac{A+A^{\dagger}}{\sqrt{2}}.
\end{equation}
\begin{proposition}\label{pro1}
The operator $Q$ in (\ref{thepos}) is an operator of the multiplication by an independent variable. That is,
\begin{equation}\label{inde}
Q\Psi_n(x)=x\Psi_n(x).
\end{equation}
\end{proposition}
\begin{proof}
Since the sequences $\{a_n\},~\{b_n\}$ and the normalized polynomials system $\{\Psi_n(x)\}$ satisfy the four term recurrence relation (\ref{nRe}), we have
\begin{eqnarray*}
Q\Psi_n(x)
&=&\frac{1}{\sqrt{2}}(A+A^{\dagger})\Psi_n(x)\\
&=&\frac{1}{\sqrt{2}}(\sqrt{2}b_{n-1}\Psi_{n-1}(x)+ \sqrt{2} a_{n}\Psi_{n}(x)  +\sqrt{2}b_n\Psi_{n+1})\\
&=&b_{n-1}\Psi_{n-1}(x)+ a_{n}\Psi_{n}(x) +b_n\Psi_{n+1}(x)\\
&=&x\Psi_n(x).
\end{eqnarray*}
\end{proof}
In this regard, unlike the algebra $\mathfrak{A}_{\ns}$, $\mathfrak{A}$ is the oscillator-like algebra most closely associated with the orthogonal polynomials system satisfying the four term recurrence relation (\ref{nRe}).
\section{Main results}
In this section we shall provide  necessary and sufficient conditions on $a_n$ and $b_n$ of the four term recurrence relation (\ref{nRe}) for the oscillator-like algebra, $\mathfrak{A}$, to be of finite dimension. The following theorem is the main result of this manuscript.
\begin{theorem}\label{t1}
The generalized oscillator-like algebra $\mathfrak{A}$ is of finite dimension if and only if
\begin{equation}\label{T1}
b^2_n=\alpha_2 n^2+\alpha_1 n+\alpha_0 \:\:and\:\: a_n=\beta_1 n+\beta_0,  \: with\:\:\alpha_1=\alpha_2+\alpha_0,
\end{equation}
where $\alpha_0, \alpha_1, \alpha_2, \beta_0, \beta_1\in\R$.
\end{theorem}
As a corollary we state the following result.
\begin{corollary}\label{c1}
If the oscillator-like algebra  $\mathfrak{A}$ is of finite dimension, then the dimension of $\mathfrak{A}$ is four.
\end{corollary}

\subsection{Proof of Theorem \ref{t1}}
We  execute the proof in three steps. In step-1 we prove that if ${\text{dim}}(\mathfrak{A})<\infty$ then $b_n$ is of second degree in $n$. In step-2 we show that if ${\text{dim}}(\mathfrak{A})<\infty$ then $a_n$ is of degree one in $n$. In step-3 we prove the converse of the theorem.\\
{\bf Step-1:}
From (\ref{wh1}) we have
\begin{equation}
[N, A^{\dagger}]\Psi_n=\af_s^{\dagger}\Psi_n,\quad n=0,1,2,\cdots.
\end{equation}
Therefore, $\af_s^{\dagger}\in \mathfrak{A}$ and hence
\begin{equation}
\mathfrak{A}_s\subseteq \mathfrak{A}.
\end{equation}
Thus, if ${\text{dim}}(\mathfrak{A})<\infty,$ then ${\text{dim}}(\mathfrak{A}_s)<\infty$. Therefore, according to \cite{Hon}, $b^2_n$ must be of the form
\begin{equation}
b^2_n=\alpha_2 n^2+\alpha_1 n+\alpha_0.
\end{equation}
{\bf Step-2:}
Since  $\af_s^{\dagger}\in \mathfrak{A},$ then  $D=A^{\dagger}-\af_s^{\dagger}\in \mathfrak{A}.$
Define a family of operators as follows:
\begin{equation}
D_1=[A^{\dagger}, D],\quad D_2=[A^{\dagger}, D_1], \quad \text{and}\quad D_k=[A^{\dagger}, D_{k-1}], \: k=3,4,\cdots
\end{equation}
and
\begin{equation}
d^{(1)}_n=a_n-a_{n-1},\:\: d^{(2)}_n=d^{(1)}_n-d^{(1)}_{n-1}\quad \text{and}\quad d^{(k)}_n=d^{(k-1)}_n-d^{(k-1)}_{n-1}, \: k=3,4,\cdots.
\end{equation}
Then, by induction, we have
\begin{equation}
D_k\Psi_n=\left(\sqrt{2}  \right)^{k+1}\prod_{i=1}^kb_{n-i}d_n^{(k)}\Psi_{n-k}.
\end{equation}
We can see that for any $\Psi_n(x),$ with $n\geq k,\: D_k$ is lowering the level of $\Psi_n(x)$ by $k-$stages if there is no $k$ such that $d_n^{(k)}=\mbox{constant}$, for $n=0,1,2,\cdots .$
Therefore, if ${\text{dim}}(  \mathfrak{A} )<\infty$ then there exists $k$ such that $d_n^{(k)}=\mbox{constant}$, for $n=0,1,2,\cdots.$
Let
\begin{equation}
p=\inf\{ k~\vert~ d_n^{(k)}=\mbox{constant},~~\mbox{for}\: n=0,1,2,\cdots\}.
\end{equation}
Then it can be easily shown that the $a_n$ has the form
\begin{equation}\label{an}
a_n=\sum_{i=0}^{p}\theta_i n^i=  \theta_{p} n^{p}+\theta_{p-1} n^{p-1}.
+...+\theta_1 n+\theta_0,
\end{equation}
where $\theta_p,\cdots,\theta_0$ are real constants. Hence, $D$ can be seen as 
\begin{equation}\label{DD}
D=\sqrt{2}\sum_{i=0}^{p}\theta_i N^i=\sqrt{2} \left( \theta_{p} N^{p}+\theta_{p-1} N^{p-1} +\cdots+\theta_1 N+\theta_0I\right)\in  \mathfrak{A}.
\end{equation}
Now let us show that ${\text{dim}}( \mathfrak{A})=\infty$ if $p\geq 2$.
Since $D,\:N,\:I\in  \mathfrak{A}$. Eq. (\ref{DD}) implies that $$\sqrt{2}(\theta_pN^{p}+\theta_{p-1} N^{p-1} +...+\theta_2 N^2)\in  \mathfrak{A}.$$
By rescaling, we get
\begin{equation}\label{W^0}
W_0= N^{p}+\gamma_{p-1} N^{p-1} +...+\gamma_2 N^2\in  \mathfrak{A},
\end{equation}
where $\gamma_i=\frac{\theta_i}{\sqrt{2}\theta_{p}},$ for $i=2,\cdots,p-1.$
The following commutation relations can easily be computed:
\begin{equation}\label{q1}
[N^2,\af_s^{\dagger}]=2\af_s^{\dagger}N+\af_s^{\dagger}
\end{equation}
\begin{equation}\label{q2}
[N^3,\af_s^{\dagger}]=3\af_s^{\dagger}N^2+3\af_s^{\dagger}N+\af_s^{\dagger}
\end{equation}
\begin{equation}\label{q3}
[N^4,\af_s^{\dagger}]=4\af_s^{\dagger}N^3+6\af_s^{\dagger}N^2+4\af_s^{\dagger}N+\af_s^{\dagger}.
\end{equation}
That is, in general we have
\begin{equation}\label{q4}
[N^k,\af_s^{\dagger}]= \sum_{i=1}^{k}c^i_k\af_s^{\dagger}N^{k-i},\quad\text{where}~~c_k^i\in\mathbb{R}.
\end{equation}
Using (\ref{W^0}) and (\ref{q4}), we have 
\begin{equation}\label{q10}
\mathcal W^{+}:= [W_0,\af_s^{\dagger}]=\sum_{i=2}^{p}\sum_{j=1}^{i}\gamma_ic^{j}_{i} \af_s^{\dagger}N^{i-j},
\end{equation}
Since $W_0\in\A$, we have $\mathcal W^{+} \in \mathfrak{A}.$ 
After $(p-2)$-iterations it can be seen that
\begin{eqnarray}\label{rq1}
\mathcal W^+_{(p-1)}&:=&\left[\af_s^{\dagger}...\left[\af_s^{\dagger},\left[\af_s^{\dagger}, \mathcal W^{+}\right] \right]... \right]\\
&=&(-1)^{p}p!\af_s^{\dagger (p-1)}N\nonumber
+f(p)\af_s^{\dagger (p-1)},
\end{eqnarray}
where $f(p)$ is some function of $p.$
Since $\mathcal W^+\in\A$, we get $\mathcal W^+_{(p-1)}\in\A$. 
Further, the following commutation relation can easily be verified by induction
\begin{equation}\label{cq2}
\left[\af_s^{\dagger m},\af_s^{\dagger (p-1)}N\right]= -m\af_s^{\dagger (p-1+m)},\:\: m\geq 1.
\end{equation}
Now, (\ref{rq1}) and (\ref{cq2}) imply that
\begin{eqnarray}\label{cq3}
\left[\af_s^{\dagger},\mathcal W^+_{(p-1)}\right]&=&(-1)^{p}p!\left[\af_s^{\dagger},\af_s^{\dagger (p-1)}N\right]\nonumber\\
&=&(-1)^{p+1}p!\af_s^{\dagger (p)}.
\end{eqnarray}
Thereby, since $\af_s^{\dagger}, \mathcal W^+_{(p-1)}\in\A$, we see $\af_s^{\dagger (p)}\in\mathfrak{A}.$
Again using the relation (\ref{cq2}), we get
\begin{eqnarray}\label{cq4}
\left[\af_s^{\dagger (p)},\mathcal W^+_{(p-1)}\right]&=&(-1)^{p}p!\left[\af_s^{\dagger (p)},\af_s^{\dagger (p-1)}N\right]\nonumber\\
&=& (-1)^{p+1}p!p\af_s^{\dagger (2p-1)},\:\: m\geq 1;
\end{eqnarray}
Thereby, $\af_s^{\dagger {(2p-1)}}  = \af_s^{\dagger {p+(p-1)}}\in\mathfrak{A}.$
By iteration, we can prove that

\begin{equation}
(\af_s^{\dagger})^{p+m(p-1)}\in \mathfrak{A},\quad \mbox{for every}\quad m=1,2,3,\cdots.
\end{equation}
Further, for $p\geq 2$, the operators $(A^{\dagger})^{p+m(p-1)}$ are new elements of $\mathfrak{A}$ for every $m\geq 1$.
Therefore $\mathfrak{A}$ is of infinite dimension. That is, we have arrived at the conclusion that if ${\text{dim}}(\mathfrak{A})<\infty$, then $p  <2.$ Thus, from (\ref{an}), $a_n=\theta_1n+\theta_0$. Hence, if ${\text{dim}}(\mathfrak{A})<\infty$, then $b_n$ and $a_n$ are second and  first degree polynomials in $n$ respectively.\\
{\bf Step-3:~~}
Let us prove that if $p<2$, then the algebra $\mathfrak{A}$ is of finite dimension. In this regard, for $n\geq 0$, assume that
\begin{equation}
b^2_n=\alpha_2 n^2+\alpha_1 n+\alpha_0 \quad and\quad a_n=\beta_1 n+\beta_0, \: with\:\:b_{-1}=0.
\end{equation}
Then from (\ref{wh1}) we have
\begin{eqnarray}\nonumber
[N, A^\dagger]\Psi_n&=&NA^\dagger \Psi_n -A^\dagger N\Psi_n\\
&=&\sqrt{2}b_n\Phi_{n+1}\\\nonumber
&=&\sqrt{2}b_n\Psi_{n+1}(x)+\sqrt{2}(\beta_1 n+\beta_0)\Psi_n-\sqrt{2}(\beta_1 n+\beta_0)\Psi_n\\\nonumber
&=&A^\dagger\Psi_n-\sqrt{2}(\beta_1 n+\beta_0)\Psi_n \\\nonumber
&=&A^\dagger\Psi_n-\sqrt{2}\beta_1N\Psi_n-\sqrt{2}\beta_0\Psi_n\\\nonumber
&=&\left(A^\dagger-\sqrt{2}\beta_1N-\sqrt{2}\beta_0I\right)\Psi_n.\nonumber
\end{eqnarray}
Therefore,
\begin{equation}
[N, A^\dagger]=A^\dagger-\sqrt{2}\beta_1N-\sqrt{2}\beta_0I.
\end{equation}
Further
\begin{eqnarray}
[N, A]\Psi_n&=&\sqrt{2}(n-1)b_{n-1}\Psi_{n-1}-\sqrt{2}nb_{n-1}\Psi_{n-1}\\\nonumber
&=& -\sqrt{2}b_{n-1}\Psi_{n-1}\\\nonumber
&=&-A\Psi_n.\nonumber
\end{eqnarray}
Therefore,
\begin{equation}
[N, A]=-A.
\end{equation}
Also
\begin{eqnarray}
[A, A^\dagger]\Psi_n&=&2\left[2\alpha_2n+(\alpha_1-\alpha_2)   \right]\Psi_{n}+2\beta_1b_{n-1}\Psi_{n-1}\\ \nonumber
&=&2\left[2\alpha_2N+(\alpha_1-\alpha_2)I   \right]\Psi_{n}+\sqrt{2}\beta_1A\Psi_{n}\\\nonumber
&=&\left[4\alpha_2N+2(\alpha_1-\alpha_2)I   +\sqrt{2}\beta_1A\right]\Psi_{n}.\nonumber
\end{eqnarray}
Therefore,
\begin{equation}
[A, A^\dagger]=4\alpha_2N+2(\alpha_1-\alpha_2)I   +\sqrt{2}\beta_1A.
\end{equation}
That is, in this case, all the commutation relations are linear combinations of the operators $A, A^{\dagger}, N$ and $I$. Therefore, the algebra $\mathfrak{A}$ is closed under the bracket $[~\cdot~,~\cdot~]$.
Hence $\mathfrak{A}$ is of finite dimension.
\subsection{Proof of Corollary \ref{c1}}
 The proof follows from step-3 of the above proof.
\section{Some examples}
In this section, as examples, we discuss the dimensions of oscillator-like algebras associated with Laguerre and Jacobi polynomials. We borrow the details of these polynomials from \cite{BD2, BD3}. For an enhanced explanation we refer the reader to \cite{BD2, BD3} and the references therein.
\subsection{Laguerre polynomials}
 The Laguerre polynomials are defined by
$$L_n^{\alpha}(x)=\frac{\alpha+1}{n!}~_1F_1(-n,\alpha+1;x).$$
These polynomials are orthogonal in the Hilbert space $\mathcal{H}_L=L^2(\R_+, x^{\alpha}e^{-x}dx)$. The normalized polynomials take the form
$$\Psi_n(x)=d_n^{-1}L_n^{\alpha}(x)\mbox{with}\quad d_n=\sqrt{\frac{\Gamma(n+\alpha+1)}{n!}};\quad n\geq 0.$$
These normalized polynomials satisfy the non-symmetric recurrence relation (\ref{nRe}) with
$$b_n=-\sqrt{(n+1)(n+\alpha+1)},\quad a_n=2n+\alpha+1.$$
Therefore, according to theorem (\ref{t1}), the related oscillator-like algebra, $\mathfrak{A}$, is of finite dimension. In the non-symmetric case, from (\ref{b-1}), (\ref{b3}) and (\ref{Com}), we can see that
$$[\widehat{\af}_{n-s}, \widehat{\af}_{n-s}^{\dagger}]=4N_{\ns}+2(\alpha+1)I_{\ns}\in\A_{\ns}.$$
Therefore all the commutators are linear combinations of $\widehat{\af}_{ns}, \widehat{\af}_{\ns}^{\dagger}, N_{\ns}$ and $I_{\ns}$. That is, the commutator operation is closed in $\A_{\ns}$. It is a finite dimensional algebra, however the algebra is not isomorphic to $\A_{WH}$. In the oscillator-like case, from (\ref{wh1}) we can see that
$$[N,A^{\dagger}]=A^{\dagger}-2\sqrt{2}N-\sqrt{2}(\alpha+1)I\in\A\quad\text{and}$$
$$[A, A^{\dagger}]=4N+2\sqrt{2}A+2(\alpha+1)I\in\A.$$
Hence the algebra $\A$ is closed under the commutator operation and therefore finite dimensional. Once again the algebra $\A$ is not isomorphic to $A_{WH}$.
\subsection{Jacobi polynomials}
The Jacobi polynomials
$$P_n^{(\alpha,\beta)}(x)=\frac{(\alpha+1)_n}{n!}~_2F_1(-n,n+\alpha+\beta; \alpha+1; \frac{1-x}{2})$$
are orthogonal in the Hilbert space 
$L^2([-1,1], (d_0(\alpha,\beta))^{-2}(1-x)^{\alpha}(1+x)^{\beta}dx)$, where
$$d_0^2(\alpha, \beta)=2^{\alpha+\beta+1}\frac{\Gamma(\alpha+1)\Gamma(\beta+1)}{\Gamma(\alpha+\beta+2)}.$$
The normalized polynomials $\{\Psi_n(x)\}_{n=0}^{\infty}$ are defined by the formula $\Psi_n(x)=d_0d_n^{-1}P_n^{(\alpha,\beta)}(x)$, where
$$d_n^2=2^{\alpha+\beta+1}\frac{\Gamma(n+\alpha+1)\Gamma(n+\beta+1)}
{\Gamma(n+\alpha+\beta+1)n!(2n+\alpha+\beta+1)};\quad n>0.$$
Then the non-symmetric recurrence relation (\ref{nRe}) is satisfied with
\begin{eqnarray*}
a_n&=& \frac{\beta^2-\alpha^2}{(2n+\alpha+\beta)(2n+\alpha+\beta+2)},\\
b_n&=&2\sqrt{\frac{(n+1)(n+1+\alpha)(n+1+\beta)(n+\alpha+\beta+1)}
{(2n+\alpha+\beta+1)(2n+\alpha+\beta+2)^2(2n+\alpha+\beta+3)}}
\end{eqnarray*}
Therefore, according to theorem (\ref{t1}), the corresponding oscillator-like algebra, $\mathfrak{A}_J$, is of infinite dimension. Here also we can compute the commutator $[A, A^{\dagger}]$ and see that the algebra does not close under the commutator operation.
\section{conclusion}
In this paper, we have discussed the dimensions of oscillator-like algebras induced by orthogonal polynomials satisfying a non-symmetric four term recurrence relation. Further, we have also responded to the claims made in \cite{BD2, BD3} about our previous paper \cite{Hon}.

In \cite{BD2} the authors have presented some remarks about the dimensions of oscillator algebras associated with two dimensional orthogonal polynomials such as the normalized 2D-Hermite polynomials $H_{n,m}(z,\overline{z})$ which satisfy the three term recurrence relation \cite{Wu1, GA}
\begin{equation}\label{He1}
zH_{m,n}(z,\overline{z})=\sqrt{m+1}H_{m+1,n}(z,\overline{z})+\sqrt{n}H_{m,n-1}(z,\overline{z}).
\end{equation}
It may be interesting to consider a detail study of the dimensions of oscillator algebras arising from 2D orthogonal polynomials satisfying three-term and four-term recurrence relations.

Further, there are several deformations to 1D and 2D orthogonal polynomials, for example see \cite{key3, qd, Bur, BH, FV, IS}. The theory developed in \cite{Hon, BD2, BD3} or in this manuscript does not directly apply to the deformed algebras associated with these deformed orthogonal polynomials. 

\end{document}